\documentclass[twoside]{IEEEtran}
\usepackage{graphicx}
\usepackage{float}
\usepackage{caption}
\usepackage{multicol}
\usepackage{mathrsfs}
\usepackage{tikz}
\usepackage{amsmath}
\usepackage{amssymb}
\usepackage{amsthm}
\usepackage{multicol}
\usepackage{mathtools}
\usepackage{float}
\usepackage{placeins}
\usepackage{epstopdf}
\usepackage{multirow}
\usepackage{subfigure}
\usepackage{afterpage}
\usepackage{pdflscape}
\usepackage{float}
\restylefloat{table}
\usepackage{enumitem}
\usepackage[utf8]{inputenc}
\floatstyle{ruled}
\newfloat{algorithm}{tbs}{loa}
\providecommand{\algorithmname}{Algorithm}
\floatname{algorithm}{\protect\algorithmname}

\theoremstyle{remark}
\newtheorem{theorem}{Theorem}

\theoremstyle{remark}
\newtheorem{example}{Example}

\title{Optimal Vector Linear Index Codes for Some Symmetric Side Information Problems}
\begin{document}
\author{Mahesh Babu Vaddi and B.~Sundar~Rajan,~\IEEEmembership{Fellow,~IEEE}
}
\maketitle
\begin{abstract}
This paper deals with vector linear index codes for multiple unicast index coding problems where there is a source with  $K$ messages and there are $K$ receivers each wanting a unique message and having symmetric (with respect to the receiver index) two-sided antidotes (side information). Optimal scalar linear index codes for several such instances of this class of problems for one-sided antidotes(not necessarily adjacent)  have been reported in \cite{MRRarXiv}. These codes can be viewed as special cases of the symmetric unicast index coding problems discussed in \cite{MCJ} with one sided adjacent antidotes.  In this paper, starting from a given multiple unicast index coding problem with with $K$ messages and one-sided adjacent antidotes for which a scalar linear index code $\mathfrak{C}$ is known, we give a construction procedure which constructs a sequence (indexed by $m$) of multiple unicast index problems with two-sided adjacent antidotes (for the same source) for all of which a vector linear code $\mathfrak{C}^{(m)}$ is obtained from $\mathfrak{C}.$  Also, it is shown that if $\mathfrak{C}$ is optimal then  $\mathfrak{C}^{(m)}$ is also optimal for all $m.$ We illustrate our construction for some of the optimal scalar linear codes of \cite{MRRarXiv} though the construction is applicable for all the codes of \cite{MRRarXiv}.\footnote{The authors are with the Department of Electrical Communication Engineering, Indian Institute of Science, Bangalore-560012, India. Email:bsrajan@ece.iisc.ernet.in.}
\end{abstract}
\section{Introduction and Background}
\label{sec1}

The problem of index coding with side information was introduced by Birk and Kol \cite{BiK} and  Bar-Yossef \textit{et al.} \cite{YBJK} studied the class of index coding problems in which each receiver demands only one single message and the number of receivers equals number of messages. Ong and Ho \cite{OnH} classify the binary index coding problem depending on the demands and the side information possessed by the receivers. An index coding problem is unicast if the demand sets of the receivers are disjoint. If the problem is unicast and if the size of each demand set is one, then it is said to be single unicast. Any unicast index problem can be equivalently reduced to an single unicast problem discussed in \cite{YBJK}. For this canonical unicast index coding problem, it was shown that the length of the optimal linear index code is equal to the minrank of the side information graph of the index coding problem but finding the minrank is NP hard. \\

Maleki \textit{et al.} \cite{MCJ} found the capacity of symmetric multiple unicast index problem with neighboring antidotes (side information). In a symmetric multiple unicast index coding problem with equal number of $K$  messages and source-destination pairs, each destination has a total of $U+D=A<K$ antidotes, corresponding to the $U$ messages before (``up" from) and $D$ messages after (``down" from) its desired message. In this setting, the $k^{th}$ receiver $R_{k}$ demands the message $x_{k}$ having the antidotes
\begin{equation}
\label{antidote}
{\cal K}_k= \{x_{k-U},\dots,x_{k-2},x_{k-1}\}\cup\{x_{k+1}, x_{k+2},\dots,x_{k+D}\}.
\end{equation}
The symmetric capacity $C$ of this index coding problem setting is shown to be as follows:
\begin{flushleft}
$U,D \in$ $\mathbb{Z},$\\
$0 \leq U \leq D$,\\
$U+D=A<K$, \\
$C=\left\{
                \begin{array}{ll}
                  {1,\qquad\quad\ A=K-1}\\
                  {\frac{U+1}{K-A+2U}},A\leq K-2\qquad $per message.$
                  \end{array}
              \right.$
\end{flushleft}
\ \\
The above expression for capacity per message can be equivalently expressed as:
\begin{equation}
\label{capacity}
C=\left\{
                \begin{array}{ll}
                  {1 ~~~~~~~~~~~~~~~~~~~~~~~~~~~~~ \mbox{if} ~~ U+D=K-1}\\
                  {\frac{min(U,D)+1}{K+min(U,D)-max(U,D)}} ~~~ \mbox{if} ~~U+D\leq K-2. 
                  \end{array}
              \right.
\end{equation}

In the setting of \cite{MCJ} with one sided antidote cases, i.e., the cases where $U$ or $D$ is zero, without loss of generality, we can assume that $max(U,D)= D$ and $min(U,D)=0$ (all the results hold when $max(U,D)=U$). In this setting, the $k^{th}$ receiver $R_{k}$ demands the message $x_{k}$ having the antidotes,
\begin{equation}
\label{antidote1}
{\cal K}_k =\{x_{k+1}, x_{k+2},\dots,x_{k+D}\}, 
\end{equation}
\noindent
for which \eqref{capacity} reduces to
\begin{equation}
\label{capacity1}
C=\left\{
                \begin{array}{ll}
                  {1 ~~~~~~~~~~~~ \mbox{if} ~~ D=K-1}\\
                  {\frac{1}{K-D}} ~~~~~~~ \mbox{if} ~~D\leq K-2 
                  \end{array}
              \right.
\end{equation}
symbols per message. \\

\subsection{Contributions}
\quad In the capacity expression given in \eqref{capacity} if $\small{U+1}$ divides $\small{K-D+U}$, capacity can be achieved by using scalar linear codes. In the scalar linear coding one packs $K$ messages in $\frac{K-D+U}{U+1}$ dimensions (code symbols). If $\small{U+1}$ does not divide $\small{K-D+U}$, vector linear coding can only achieve capacity. In the vector linear code, one needs to pack $K(U+1)$ message symbols corresponding to $K$ users in $K-D+U$ dimensions. \\

In \cite{MCJ} Maleki $et\ al.$ proved that vector linear coding exists for any arbitrary $U$ and $D$ over a sufficiently large field size by imposing conditions on encoding and decoding matrices $U_{m,k}$ and $V_{m}$. In Section II we prove that vector linear solution exists for a two-sided symmetric antidote problem with $U$ antidotes above and $D$ antidotes below if a scalar linear solution exists for one-sided antidote problem with same number of messages and number of one-sided antidotes $\Delta =\vert D-U \vert$. 
We give a  construction procedure which constructs a sequence of multiple unicast index problems with two-sided antidotes starting from a given multiple unicast index coding problem of one-sided antidote. It is shown that if there is an optimal scalar linear index code for the starting problem then this code can be used to construct an optimal vector linear index code for all the extended problems.  In \cite{MRRarXiv} the authors proposed optimal scalar linear index codes for ten classes of one sided (not necessarily adjacent) symmetric multiple unicast index coding problems with optimal scalar linear index codes. The antidotes assumed in this work is a proper subset of \eqref{antidote1} for eight classes. These optimal codes continue to be optimal if the antidotes are taken to be adjacent as given in \eqref{antidote1}. We illustrate our construction to some of the ten cases of the symmetric multiple unicast problems studied in \cite{MRRarXiv} and demonstrate the new classes of symmetric multicast problems created for all of which an optimal vector linear code is exhibited. \\

In \cite{BVRarXiv}, we proposed a lifting construction which constructs a sequence of multiple unicast index problems with one-sided antidotes with a scalar linear index code starting from a given multiple unicast index coding problem with a known scalar linear index code. The construction in this paper is different in the following two respects:
\begin{enumerate}
\item The construction in \cite{BVRarXiv} starts from a problem with $K$ messages and gives a sequence of problems with $mK$ number of messages for $m=2,3, \dots,$ i.e., the number of messages goes on increasing as the index $m$ moves. Whereas in this work the number of messages remains $K$ and only the size of the antidote sets increase.
\item In \cite{BVRarXiv} new scalar linear index codes are obtained starting from a scalar linear index code for problems of different source sizes whereas in this paper new vector linear index codes are obtained starting from a scalar linear index code for the problem with the same source size. \\ 

\item The lifting construction in \cite{BVRarXiv} results in index coding problems with one sided antidotes where as the construction in this paper results in index coding problems with two sided antidotes.

\end{enumerate} 
Throughout the paper  WLOG we consider the case  $D \geq U$. All the codes discussed in this paper also applicable for $U \geq D$. The decoding procedure in this paper is considered for the binary field. However the index codes considered in this paper works for any finite field.

\section{extension of scalar linear code into vector linear codes}

Let the messages symbols be $\{x_1, x_2, \dots x_K \}.$  For every $x_k,$ when we deal with vector linear index codes the different messages symbols corresponding to $x_k$ are denoted by $x_{k,1}, x_{k,2}, x_{k,3}, \dots$ etc.  

If $U+1$ do not divide $K-D+U$, scalar linear codes can not achieve capacity. In this case capacity can be achieved by using vector linear coding. In the vector linear code, we require to pack $K(U+1)$ message symbols corresponding to $K$ users in $K-D+U=K-\Delta$ dimensions. 

We prove that we can pack $K(U+1)$ message symbols into $K$ symbols $\{y_{1},y_{2},\cdots,y_{K}\}$ and then convert these $K$ symbols into $K-\Delta$ code symbols by using one side adjacent antidote scalar linear code. Define the symbol $y_{k}$ for $k=1,2,\cdots,K$ as 
\begin{equation}
y_{k}=x_{k,1}+x_{k-1,2},...,x_{k-U,U}+x_{k-U,U+1}.
\end{equation}
The symbol $y_{k}$ comprises of  $(U+1)$ message symbols and each of the $K(U+1)$ message symbols appear exactly once in one of the $y_{k}$. In the symbol  $y_{k}$, there exists $(U+1)$ message symbols and these $(U+1)$ message symbols are required by the receivers $R_{j}$, for $ j =k,k-1,\cdots, k-U.$   \\

\begin{figure}[htbp]
\centering{}
\includegraphics[scale=0.9]{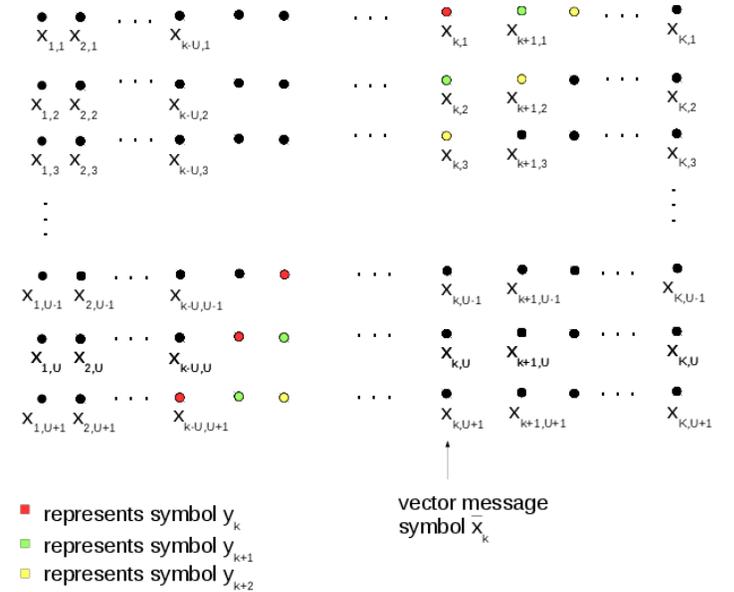}
\label{fig:IFGraphExampleTheorem}
\caption{\small Vector coding message alignment.}
\end{figure}
\begin{theorem}
\label{thm1}
For a multiple unicast index coding problem with $K$ messages $\{y_1,y_2,\cdots,y_K\}$ and the same number of receivers with the receiver $R_k$ wanting the message $y_k$  and having  a symmetric antidote pattern ${\cal K}_k$ given by 
\begin{equation}
\label{antidoteext}
{\cal K}_k =\{y_{k+1}, y_{k+2},\dots,y_{k+D}\},
\end{equation} 
let $\mathfrak{C} = \{ t_1,t_2, \cdots, t_l \}$ be a scalar linear code of  length $l.$ Define $\overline{x_k}=\{x_{k,1},x_{k,2},\cdots,x_{k,U+1}\}$ as $k^{th}$ vector message symbol for $k=1,2,\cdots,K$. For an arbitrary positive integer $U$  consider the index coding problem with  $K$ number of messages $\{\overline{x_1},\overline{x_2},\cdots,\overline{x_{K}}\}$ and the number of receivers being $K,$ and receiver $R_k$ ($k=1,2, \cdots ,K$) having antidote pattern given by 
\begin{equation}
\label{antidoteextension}
{\cal K}_k =\{\overline{x_{k-U}},\dots,\overline{x_{k-2}},\overline{x_{k-1}}\}~\cup~\{\overline{x_{k+1}}, \overline{x_{k+2}},\dots,\overline{x_{k+D+U}}\}.
\end{equation}
For this index coding problem the  code $\mathfrak{C}^{(U+1)}$ is obtained by replacing every message symbol $y_k$ in the code symbols of $\mathfrak{C}$ with $\sum_{i=1}^{U+1} x_{k+1-i,i}$ for $1 \leq k \leq K,$ i.e., by making the substitution $y_k= \sum_{i=1}^{U+1} x_{k+1-i,i}.$
\end{theorem}
\begin{proof}
The given scalar linear code $\mathfrak{C} = \{ t_1,t_2, \cdots, t_l \}$ of length  $l$ for a multiple unicast index coding problem with $K$ messages $\{y_1,y_2,\cdots,y_K\}$ and the same number of receivers with the receiver $R_k$ wanting the message $y_k$  and having  a symmetric antidote pattern ${\cal K}_k$ given in \eqref{antidoteext} enables the decoding of $K$ messages $\{y_1,y_2,\cdots,y_K\}$. That is, with linear decoding there exist  combination of code symbols $\{ t_1,t_2, \cdots, t_l \}$ to get the sum of the form given in \eqref{sum} 
\begin{equation}
\label{sum}
S_{k}=y_{k}+y_{k+a_{k,1}}+y_{k+a_{k,2}}+\cdots+y_{k+a_{k,d}} \ \vert \ k=1,2,\cdots,K, \\
\end{equation}
\noindent
for $1 \leq a_{k,1}<a_{k,2}< \cdots <a_{k,d} \leq D,$ such that $y_{k+a_{k,1}}, y_{k+a_{k,2}},\cdots,y_{k+a_{k,d}}$ are in antidotes of receiver $R_{k}$.  Since $a_{k,1},a_{k,2},\cdots,a_{k,d}$ are  less than or equal to $D$ for $k=1,2,\cdots,K$  from the sum $y_{k}+y_{k+a_{k,1}}+y_{k+a_{k,2}}+\cdots+y_{k+a_{k,d}}$, receiver $R_{k}$ can decode its wanted message $y_{k}$.

Given an arbitrary positive integer $U$  consider the index coding problem with  $K$ number of vector message symbols $\{\overline{x_1},\overline{x_2},\cdots,\overline{x_{K}}\}$ and the number of receivers being $K,$ and the receiver $R_k$ ($k=1,2, \cdots ,K$) having antidote pattern given by \eqref{antidoteextension}. The code $\mathfrak{C}^{(U+1)}$ is the index code obtained by making the substitution $y_k= \sum_{i=1}^{U+1} x_{k+1-i,i}$ in the available code $\mathfrak{C}$. We prove that the receiver $R_{k}$ can decode its wanted message set $\overline{x_k}$ ($U+1$ message symbols $x_{k,1},x_{k,2},\cdots,x_{k,U+1}$) by using the code $\mathfrak{C}^{(U+1)}$. For the code $\mathfrak{C}^{(U+1)}$, the sum in \eqref{sum} can be written by \\

{\scriptsize
\begin{equation}
\label{sumextension}
S_{k}^{(U+1)}=\sum\limits_{i=1}^{U+1} x_{k+1-i,i}+\sum\limits_{i=1}^{U+1} x_{k+a_{k,1}+1-i,i}+\cdots+\sum\limits_{i=1}^{U+1} x_{k+a_{k,d}+1-i,i}
\end{equation}
}
We describe the entire decoding process in the following three steps. The Fig.1 will be helpful to trace these steps.\\

\noindent
\emph{Step 1. Decoding of $(U+1)^{th}$ message symbol by receiver $R_{k}$}
\\
\\
The sum $S_{k+U}^{(U+1)}$ similar to  \eqref{sumextension} can be written as 
\begin{small}
$S_{k+U}^{(U+1)}=x_{k,U+1}+\sum_{i=1}^{U} x_{k+U+1-i,i}+\sum_{i=1}^{U+1} x_{k+U+a_{k,1}+1-i,i}+\cdots+\sum_{i=1}^{U+1} x_{k+U+a_{k,d}+1-i,i}.$
\end{small}

In $S_{k+U}^{(U+1)}$, only one message symbol is present which is required by $R_{k}.$ We have  $1 \leq a_{k,1}<a_{k,2}< \cdots <a_{k,d} \leq D$ and thus all other message symbols present in $S_{k+U}^{(U+1)}$ are in antidotes of receiver $R_{k}$ according to antidote pattern as in \eqref{antidoteextension}. (Note that if $S_{k+U}^{U+1}$ comprise $\geq 2$ message symbols which belong to  $\overline{x_{k}}$ and $R_{k}$ requires at least two of them, then messages interfere and $R_{k}$ can not decode any of them). Thus by using $S_{k+U}^{(U+1)}$, receiver $R_{k}$ can decode its $(U+1)^{th}$  message symbol $x_{k,U+1}$. The code $\mathfrak{C}$ enables the decoding of $\{y_{1},y_{2},\cdots,y_{K}\}$ from the sums $S_{1},S_{2},\cdots,S_{K}$, which implies that code $\mathfrak{C}^{(U+1)}$ enables the the decoding of $\{x_{1,U+1},x_{2,U+1},\cdots,x_{k-U,U+1},\cdots,x_{K,U+1}\}.$
\\
\\
\noindent
\emph{Step 2. Decoding of $U^{th}$ message symbol by receiver $R_{k}$}
\\
\\
Receiver $R_{k}$ uses sum $S_{k+U-1}^{(U+1)}$ to decode its $U^{th}$ message symbol $x_{k,U}$. The sum $S_{k+U-1}^{(U+1)}$ similar to \eqref{sumextension} can be written as 
\\
\begin{small}
$S_{k+U-1}^{(U+1)}=x_{k,U}+a(x_{k,U+1})+\sum_{i=1}^{U} x_{k+U-i,i}+\sum_{i=1}^{U+1} x_{k+U+a_{k,1}-i,i}+\cdots+\sum_{i=1}^{U+1} x_{k+U+a_{k,d}-i,i},$ 
\end{small} 
\\
where $a=1$ if $a_{k,1}=1$, else $a=0$. That is, depending on the value of $a_{k,1}$, sum $S_{k+U-1}^{(U+1)}$ comprises of either one message symbol or two message symbols belongs to vector message symbol $\overline{x_{k}}$. \\

If $a_{k,1}>1$,  $S_{k+U-1}^{(U+1)}$ comprises only one message symbol ($x_{k,U}$) which belongs to $\overline{x_{k,U}}$. All other messages present in $S_{k+U-1}^{(U+1)}$ are antidotes to $R_{k}$. Thus receiver $R_{k}$ can decode message symbol $x_{k,U}$. \\

If $a_{k,1}=1$, $S_{k+U-1}^{(U+1)}$ comprises two message symbols ($x_{k,U+1},x_{k,U}$) belongs to vector message symbol $\overline{x_{k}}.$ Receiver $R_{k}$ has already decoded the message symbol $x_{k,U+1}$ and all other messages present in $S_{k+U-1}^{(U+1)}$ are antidotes to $R_{k}$. Thus receiver $R_{k}$ can decode message symbol $x_{k,U}$. 
\\
\\
\noindent
\emph{Step l. Decoding of $(U+2-l)^{th} (1 \leq l \leq U+1)$ message symbol by receiver $R_{k}$}
\\
\\
Receiver $R_{k}$ uses sum $S_{k+U+1-l}^{(U+1)}$ to decode its $(U+2-l)^{th}$ message symbol $x_{k,U+2-l}$. The sum $S_{k+U+1-l}^{(U+1)}$ similar to \eqref{sumextension} can be written as  \\
\begin{small}
\begin{equation}
\begin{aligned}
S_{k+U+1-l}^{(U+1)}=x_{k,U+2-l}&+1_{A_{1}}(l)x_{k,U+2-l+a_{k,1}} \\
&+1_{A_{2}}(l)x_{k,U+2-l+a_{k,2}}+ \\
\cdots &+ 1_{A_{d}}(l)x_{k,U+2-l+a_{k,d}}\\
&+\sum\limits_{i=1, i\neq U+2-l}^{U+1} x_{k+U+2-l-i,i}\\&+
\sum\limits_{i=1,i \neq U+2-l+a_{k,1}}^{U+1}x_{k+U+2-l+a_{k,1}-i,i}+
\cdots \\&+\sum\limits_{i=1, i \neq U+2-l+a_{k,d}}^{U+1} x_{k+U+2-l+a_{k,d}-i,i}.
\end{aligned}
\end{equation}
\end{small}

Where $1_{A}$ is the indicator function such that $1_{A}(x)=1$  if $x \in A$, else it is zero. $A_{j}=\{a_{k,j}+1,a_{k,j}+2,\cdots,U+1\}$ if $U \geq a_{k,j}$, else $A_{j}=\{\Phi\}$ for $j=1,2,\cdots,d$. In the above sum $x_{k,U+2-l}$ is the required message. \\

The quantity $1_{A_{j}}(l)x_{k,U+2-l+a_{k,j}}$ for $j=1,2,\cdots,d$ is the interference to the receiver $R_{k}$ from its wanted vector message symbol $\overline{x_{k}}$. Receiver $R_{k}$ already knows the $l-1$ message symbols $\{x_{k,U+1}, x_{k,U},\cdots,x_{k,U+3-l}\}$ and thus the interference from the wanted message symbol $\overline{x_{k}}$ in the sum $S_{k+U+1-l}^{(U+1)}$ can be canceled. \\

The message symbols  in the terms $\tiny{\sum_{i=1, i\neq U+2-l+a_{k,j}}^{U+1} x_{k+U+2-l+a_{k,j}-i,i}}$ for $j=1,2,\cdots,d$ is in antidotes for receiver $R_{k}$. Thus receiver $R_{k}$ can decode $x_{k,U+2-l}$ from $S_{k+U+1-l}^{(U+1)}$. \\

This procedure continued until all receiver $R_{k}$ decode its $U+1$ wanted messages $\{x_{k,1},x_{k,2},\cdots,x_{k,U+1}\}$. Thus receiver $R_{k}$ decodes its wanted vector message symbol $\overline{x_{k}}$. This completes the proof for decoding. 
\end{proof}
\begin{theorem}
\label{thm2}
In the construction of Theorem \ref{thm1}, if the given code $\mathfrak{C}$ has optimal length $l=K-D,$ then all extended vector linear codes of given code are of optimal length and hence capacity achieving.
\end{theorem}
\begin{proof}
The number of code symbols in the extended code is equal to the number of code symbols in $\mathfrak{C}$ which is equal to $K-\Delta$ = $K-D+U$. By using $K-D+U$ code symbols, every receiver gets $U+1$ of its wanted messages. The capacity achieved by this code is $\frac{U+1}{K-D+U}$ per message, which is equal to the capacity mentioned in \eqref{capacity}. It follows that the extended vector linear codes are of optimal length.
\end{proof}
\medskip

The following examples illustrates Theorem \ref{thm1} and Theorem \ref{thm2}
\begin{example}
$K=20, \ U=0, \ D=4.$
\\
\\
We have one-sided antidote scalar code $\mathfrak{C}$=$\{y_{1}+y_{5}, ~y_{2}+y_{6}, ~y_{3}+y_{7}, ~y_{4}+y_{8}, y_{5}+y_{9}, ~y_{6}+y_{10}, ~y_{7}+y_{11}, ~y_{8}+y_{12}, ~y_{9}+y_{13}, ~y_{10}+y_{14}, ~y_{11}+y_{15}, ~y_{12}+y_{16}, ~y_{13}+y_{17}, ~y_{14}+y_{18}, ~y_{15}+y_{19}, ~y_{16}+y_{20}\}$.
\\
\\
$Case~ I$:  $K=20,\ U=1,\ D=5.$
\\
\\
$\Delta=4$, Capacity=$\frac{2}{16}$. Define $y_{i}=x_{i,1}+x_{i-1,\ 2}$ for $i=1,2,\dots,20.$ We have
\\
\\
\begin{small}
$y_{1}=x_{1,1}+x_{20,2}, ~~ y_{11}=x_{11,1}+x_{10,2}$,\\
$y_{2}=x_{2,1}+x_{1,2},  ~~~y_{12}=x_{12,1}+x_{11,2}$,\\
$y_{3}=x_{3,1}+x_{2,2},  ~~~~y_{13}=x_{13,1}+x_{12,2}$,\\
$y_{4}=x_{4,1}+x_{3,2},  ~~~~y_{14}=x_{14,1}+x_{13,2}$,\\
$y_{5}=x_{5,1}+x_{4,2},  ~~~~y_{15}=x_{15,1}+x_{14,2}$,\\
$y_{6}=x_{6,1}+x_{5,2},  ~~~~y_{16}=x_{16,1}+x_{15,2}$,\\
$y_{7}=x_{7,1}+x_{6,2},  ~~~~y_{17}=x_{17,1}+x_{16,2}$,\\
$y_{8}=x_{8,1}+x_{7,2},  ~~~~y_{18}=x_{18,1}+x_{17,2}$,\\
$y_{9}=x_{9,1}+x_{8,2},  ~~~~y_{19}=x_{19,1}+x_{18,2}$,\\
$y_{10}=x_{10,1}+x_{9,2},  ~~y_{20}=x_{20,1}+x_{19,2}$.
\\
We get the code $\mathfrak{C^{(2)}}$ given by
\\
 $\mathfrak{C^{(2)}}=\{{x_{1,1}}+x_{20,2}+x_{5,1}+x_{4,2},\
 {x_{2,1}}+x_{1,2}+x_{6,1}+x_{5,2},\
 {x_{3,1}}+x_{2,2}+x_{7,1}+x_{6,2},\
 {x_{4,1}}+x_{3,2}+x_{8,1}+x_{7,2},\
 {x_{5,1}}+x_{4,2}+x_{9,1}+x_{8,2},\
 {x_{6,1}}+x_{5,2}+x_{10,1}+x_{9,2},\
 {x_{7,1}}+x_{6,2}+x_{11,1}+x_{10,2},\
 {x_{8,1}}+x_{7,2}+x_{12,1}+x_{11,2},\
 {x_{9,1}}+x_{8,2}+x_{13,1}+x_{12,2},\
 {x_{10,1}}+x_{9,2}+x_{14,1}+x_{13,2},\
 {x_{11,1}}+x_{10,2}+x_{15,1}+x_{14,2},\
 {x_{12,1}}+x_{11,2}+x_{16,1}+x_{15,2},\
 {x_{13,1}}+x_{12,2}+x_{17,1}+x_{16,2},\
 {x_{14,1}}+x_{13,2}+x_{18,1}+x_{17,2},\
 {x_{15,1}}+x_{14,2}+x_{19,1}+x_{18,2},\
 {x_{16,1}}+x_{15,2}+x_{20,1}+x_{19,2}\}$.
\end{small}
\\
\\
$Case~ II$:  $K=20,\ U=2,\ D=6.$
\\
\\
$\Delta=4$, Capacity=$\frac{3}{16}$. Defining $y_{i}=x_{i,1}+x_{i-1,\ 2}+x_{i-2,\ 3}$ for $i=1,2,\dots,20$ we get
\\
\\
\begin{small}
$y_{1}=x_{1,1}+x_{20,2}+x_{19,3}, ~~ y_{11}=x_{11,1}+x_{10,2}+x_{9,3}$,\\
$y_{2}=x_{2,1}+x_{1,2}+x_{20,3},  ~~~y_{12}=x_{12,1}+x_{11,2}+x_{10,3}$,\\
$y_{3}=x_{3,1}+x_{2,2}+x_{1,3},  ~~~~y_{13}=x_{13,1}+x_{12,2}+x_{11,3}$,\\
$y_{4}=x_{4,1}+x_{3,2}+x_{2,3},  ~~~~y_{14}=x_{14,1}+x_{13,2}+x_{12,3}$,\\
$y_{5}=x_{5,1}+x_{4,2}+x_{3,3},  ~~~~y_{15}=x_{15,1}+x_{14,2}+x_{13,3}$,\\
$y_{6}=x_{6,1}+x_{5,2}+x_{4,3},  ~~~~y_{16}=x_{16,1}+x_{15,2}+x_{14,3}$,\\
$y_{7}=x_{7,1}+x_{6,2}+x_{5,3},  ~~~~y_{17}=x_{17,1}+x_{16,2}+x_{15,3}$,\\
$y_{8}=x_{8,1}+x_{7,2}+x_{6,3},  ~~~~y_{18}=x_{18,1}+x_{17,2}+x_{16,3}$,\\
$y_{9}=x_{9,1}+x_{8,2}+x_{7,3},  ~~~~y_{19}=x_{19,1}+x_{18,2}+x_{17,3}$,\\
$y_{10}=x_{10,1}+x_{9,2}+x_{8,3},  ~~y_{20}=x_{20,1}+x_{19,2}+x_{18,3}$.
\\
For this case we get the extended code $ \mathfrak{C^{(3)}}$ to be
\\
 $ \mathfrak{C^{(3)}}=\{x_{1,1}+x_{20,2}+x_{19,3}+x_{5,1}+x_{4,2}+x_{3,3}, 
~~ {x_{2,1}}+x_{1,2}+x_{20,3}+x_{6,1}+x_{5,2}+x_{4,3}, 
~~ {x_{3,1}}+x_{2,2}+x_{1,3}+x_{7,1}+x_{6,2}+x_{5,3}, 
~~ {x_{4,1}}+x_{3,2}+x_{2,3}+x_{8,1}+x_{7,2}+x_{6,3}, 
~~ {x_{5,1}}+x_{4,2}+x_{3,3}+x_{9,1}+x_{8,2}+x_{7,3}, 
~~ {x_{6,1}}+x_{5,2}+x_{4,3}+x_{10,1}+x_{9,2}+x_{8,3}, 
~~ {x_{7,1}}+x_{6,2}+x_{5,3}+x_{11,1}+x_{10,2}+x_{9,3}, 
~~ {x_{8,1}}+x_{7,2}+x_{6,3}+x_{12,1}+x_{11,2}+x_{10,3}, 
~~ {x_{9,1}}+x_{8,2}+x_{7,3}+x_{13,1}+x_{12,2}+x_{11,3}, 
~~ {x_{10,1}}+x_{9,2}+x_{8,3}+x_{14,1}+x_{13,2}+x_{12,3}, 
~~ {x_{11,1}}+x_{10,2}+x_{9,3}+x_{15,1}+x_{14,2}+x_{13,3}, 
~~ {x_{12,1}}+x_{11,2}+x_{10,3}+x_{16,1}+x_{15,2}+x_{14,3}, 
~~ {x_{13,1}}+x_{12,2}+x_{11,3}+x_{17,1}+x_{16,2}+x_{15,3}, 
~~ {x_{14,1}}+x_{13,2}+x_{12,3}+x_{18,1}+x_{17,2}+x_{16,3}, 
~~ {x_{15,1}}+x_{14,2}+x_{13,3}+x_{19,1}+x_{18,2}+x_{17,3}, 
~~ {x_{16,1}}+x_{15,2}+x_{14,3}+x_{20,1}+x_{19,2}+x_{18,3}\}$.
\end{small}
\\
\\
$Case ~ III$:  $K=20,\ U=3,\ D=7.$
\\
\\
$\Delta=4$, Capacity=$\frac{3}{16}$.  With defining $y_{i}=x_{i,1}+x_{i-1,\ 2}+x_{i-2,\ 3}+x_{i-3,\ 4}$ for $i=1,2,\dots,20,$ we end with
\\
\\
\begin{small}
$y_{1}=x_{1,1}+x_{20,2}+x_{19,3}+x_{18,4}, ~~ y_{11}=x_{11,1}+x_{10,2}+x_{9,3}+x_{8,4}$,\\
$y_{2}=x_{2,1}+x_{1,2}+x_{20,3}+x_{19,4},  ~~~y_{12}=x_{12,1}+x_{11,2}+x_{10,3}+x_{9,4}$,\\
$y_{3}=x_{3,1}+x_{2,2}+x_{1,3}+x_{20,4},  ~~~~y_{13}=x_{13,1}+x_{12,2}+x_{11,3}+x_{10,4}$,\\
$y_{4}=x_{4,1}+x_{3,2}+x_{2,3}+x_{1,4},  ~~~~y_{14}=x_{14,1}+x_{13,2}+x_{12,3}+x_{11,4}$,\\
$y_{5}=x_{5,1}+x_{4,2}+x_{3,3}+x_{2,4},  ~~~~y_{15}=x_{15,1}+x_{14,2}+x_{13,3}+x_{12,4}$,\\
$y_{6}=x_{6,1}+x_{5,2}+x_{4,3}+x_{3,4},  ~~~~y_{16}=x_{16,1}+x_{15,2}+x_{14,3}+x_{13,4}$,\\
$y_{7}=x_{7,1}+x_{6,2}+x_{5,3}+x_{4,4},  ~~~~y_{17}=x_{17,1}+x_{16,2}+x_{15,3}+x_{14,4}$,\\
$y_{8}=x_{8,1}+x_{7,2}+x_{6,3}+x_{5,4},  ~~~~y_{18}=x_{18,1}+x_{17,2}+x_{16,3}+x_{15,4}$,\\
$y_{9}=x_{9,1}+x_{8,2}+x_{7,3}+x_{6,4},  ~~~~y_{19}=x_{19,1}+x_{18,2}+x_{17,3}+x_{16,4}$,\\
$y_{10}=x_{10,1}+x_{9,2}+x_{8,3}+x_{7,4},  ~~y_{20}=x_{20,1}+x_{19,2}+x_{18,3}+x_{17,4}$.
\\
The resulting extended vector linear code is 
\\
 $ \mathfrak{C^{(4)}}=\{x_{1,1}+x_{20,2}+x_{19,3}+x_{18,4}+x_{5,1}+x_{4,2}+x_{3,3}+x_{2,4}, 
~~ {x_{2,1}}+x_{1,2}+x_{20,3}+x_{19,4}+x_{6,1}+x_{5,2}+x_{4,3}+x_{3,4}, 
~~ {x_{3,1}}+x_{2,2}+x_{1,3}+x_{20,4}+x_{7,1}+x_{6,2}+x_{5,3}+x_{4,4}, 
~~ {x_{4,1}}+x_{3,2}+x_{2,3}+x_{1,4}+x_{8,1}+x_{7,2}+x_{6,3}+x_{5,4}, 
~~ {x_{5,1}}+x_{4,2}+x_{3,3}+x_{2,4}+x_{9,1}+x_{8,2}+x_{7,3}+x_{6,4}, 
~~ {x_{6,1}}+x_{5,2}+x_{4,3}+x_{3,4}+x_{10,1}+x_{9,2}+x_{8,3}+x_{7,4}, 
~~ {x_{7,1}}+x_{6,2}+x_{5,3}+x_{4,4}+x_{11,1}+x_{10,2}+x_{9,3}+x_{8,4}, 
~~ {x_{8,1}}+x_{7,2}+x_{6,3}+x_{5,4}+x_{12,1}+x_{11,2}+x_{10,3}+x_{9,4}, 
~~ {x_{9,1}}+x_{8,2}+x_{7,3}+x_{6,4}+x_{13,1}+x_{12,2}+x_{11,3}+x_{10,4}, 
~~ {x_{10,1}}+x_{9,2}+x_{8,3}+x_{7,4}+x_{14,1}+x_{13,2}+x_{12,3}+x_{11,4}, 
~~ {x_{11,1}}+x_{10,2}+x_{9,3}+x_{8,4}+x_{15,1}+x_{14,2}+x_{13,3}+x_{12,4}, 
~~ {x_{12,1}}+x_{11,2}+x_{10,3}+x_{9,4}+x_{16,1}+x_{15,2}+x_{14,3}+x_{13,4}, 
~~ {x_{13,1}}+x_{12,2}+x_{11,3}+x_{10,4}+x_{17,1}+x_{16,2}+x_{15,3}+x_{18,4},
~~ {x_{14,1}}+x_{13,2}+x_{12,3}+x_{11,4}+x_{18,1}+x_{17,2}+x_{16,3}+x_{15,4}, 
~~ {x_{15,1}}+x_{14,2}+x_{13,3}+x_{12,4}+x_{19,1}+x_{18,2}+x_{17,3}+x_{16,4}, 
~~ {x_{16,1}}+x_{15,2}+x_{14,3}+x_{13,4}+x_{20,1}+x_{19,2}+x_{18,3}+x_{17,4}\}$.
\end{small}
\end{example}
\medskip

\begin{theorem}
\label{thm3}
Consider a multiple unicast index coding problem with $K$ messages and the same number of receivers with the receiver $R_k$ wanting the message $x_k$  and having  a symmetric antidote pattern ${\cal K}_k$ given in \eqref{antidote}. For this index coding problem vector linear solution exists if scalar linear solution exists for one-sided antidote problem as mentioned in \eqref{antidote1} with same number of messages and number of one-sided antidotes $\Delta =\vert D-U \vert$.  The vector linear index code is optimal if the scalar linear code is optimal.\\
\end{theorem}
\begin{proof}
Proof follows from Theorem \ref{thm1} and Theorem \ref{thm2}.
\end{proof}
\medskip
\begin{theorem}
\label{thm4}
Proposed codes in \cite{MRRarXiv} can be extended for all $U$ and $D$ with $\Delta$ = $max(U,D)-min(U,D)$ for the following problem instances:
\begin{enumerate}
\item $\Delta$ divides $K$
\item $K-\Delta$ divides $K$
\item $\frac{K}{2}-\Delta$ divides $\Delta$ 
\item $\Delta - \frac{K}{2}$ divides $\frac{K}{2}$ 
\item $\Delta$ divides $K-\lambda$ and $\lambda$ divides $\Delta$ where $\lambda$ is an integer
\item $K-\Delta$ divides $K-\lambda$ and $\lambda$ divides $K-\Delta$ 
\item $\Delta+\lambda$ divides $K$ and $\lambda$ divides $\Delta$
\item $K-\Delta+\lambda$ divides $K$ and $\lambda$ divides $K-\Delta$
\item $\Delta$ divides $K+\lambda$ and $\lambda$ divides $\Delta$
\item $K-\Delta$ divides $K+\lambda$ and $\lambda$ divides $K-\Delta$
\end{enumerate}
\end{theorem}
\begin{proof}
In \cite{MRRarXiv} codes for symmetric instances of one side antidote problem for a given $K$ and $D$ satisfying the above mentioned conditions with $\Delta$ replaced with $D$ were proposed. The proposed codes in \cite{MRRarXiv} can also be used  for the one-sided adjacent antidotes as given in \eqref{antidote1} and these codes achieve the capacity in \eqref{capacity1}. Then the  proof follows from Theorem \ref{thm1}.
\end{proof}
\medskip
\emph{Corollary 1.} If $D$ divides $K$ then the proposed  optimal length scalar linear code is 
$\mathfrak{C}=\{x_{i+(j-1)D}+x_{i+jD}|~i=1,2,\dots,D,~ j=1,2,\dots, \frac{K}{D}-1\}$ 
\begin{example}
$K=20,\ U=2,\ D=6.$
\\
\\
$K=20, \Delta=4$, capacity=$\frac{3}{16}$. \\
Let $y_{i}=x_{i,1}+x_{i-1,\ 2}+x_{i-2,\ 3}$ for $i=1,2,\dots,20.$
\\
\\
\begin{small}
$y_{1}=x_{1,1}+x_{20,2}+x_{19,3}, ~~ y_{11}=x_{11,1}+x_{10,2}+x_{9,3}$,\\
$y_{2}=x_{2,1}+x_{1,2}+x_{20,3},  ~~~y_{12}=x_{12,1}+x_{11,2}+x_{10,3}$,\\
$y_{3}=x_{3,1}+x_{2,2}+x_{1,3},  ~~~~y_{13}=x_{13,1}+x_{12,2}+x_{11,3}$,\\
$y_{4}=x_{4,1}+x_{3,2}+x_{2,3},  ~~~~y_{14}=x_{14,1}+x_{13,2}+x_{12,3}$,\\
$y_{5}=x_{5,1}+x_{4,2}+x_{3,3},  ~~~~y_{15}=x_{15,1}+x_{14,2}+x_{13,3}$,\\
$y_{6}=x_{6,1}+x_{5,2}+x_{4,3},  ~~~~y_{16}=x_{16,1}+x_{15,2}+x_{14,3}$,\\
$y_{7}=x_{7,1}+x_{6,2}+x_{5,3},  ~~~~y_{17}=x_{17,1}+x_{16,2}+x_{15,3}$,\\
$y_{8}=x_{8,1}+x_{7,2}+x_{6,3},  ~~~~y_{18}=x_{18,1}+x_{17,2}+x_{16,3}$,\\
$y_{9}=x_{9,1}+x_{8,2}+x_{7,3},  ~~~~y_{19}=x_{19,1}+x_{18,2}+x_{17,3}$,\\
$y_{10}=x_{10,1}+x_{9,2}+x_{8,3},  ~~y_{20}=x_{20,1}+x_{19,2}+x_{18,3}$.\\
\medskip

The proposed code is $\mathfrak{C}=\{y_{1}+y_{5}, ~y_{2}+y_{6}, ~y_{3}+y_{7}, ~y_{4}+y_{8}, ~y_{5}+y_{9}, ~y_{6}+y_{10}, ~y_{7}+y_{11}, ~y_{8}+y_{12}, ~y_{9}+y_{13}, ~y_{10}+y_{14}, ~y_{11}+y_{15}, ~y_{12}+y_{16}, ~y_{13}+y_{17}, ~y_{14}+y_{18}, ~y_{15}+y_{19}, ~y_{16}+y_{20}\}$.
\medskip

 $ ~ \mathfrak{C^{(3)}}=\{x_{1,1}+x_{20,2}+x_{19,3}+x_{5,1}+x_{4,2}+x_{3,3},~~{x_{2,1}}+x_{1,2}+x_{20,3}+x_{6,1}+x_{5,2}+x_{4,3},~~ {x_{3,1}}+x_{2,2}+x_{1,3}+x_{7,1}+x_{6,2}+x_{5,3},~~ {x_{4,1}}+x_{3,2}+x_{2,3}+x_{8,1}+x_{7,2}+x_{6,3},~~ {x_{5,1}}+x_{4,2}+x_{3,3}+x_{9,1}+x_{8,2}+x_{7,3},~~ {x_{6,1}}+x_{5,2}+x_{4,3}+x_{10,1}+x_{9,2}+x_{8,3},~~ {x_{7,1}}+x_{6,2}+x_{5,3}+x_{11,1}+x_{10,2}+x_{9,3},~~ {x_{8,1}}+x_{7,2}+x_{6,3}+x_{12,1}+x_{11,2}+x_{10,3},~~ {x_{9,1}}+x_{8,2}+x_{7,3}+x_{13,1}+x_{12,2}+x_{11,3},~~ {x_{10,1}}+x_{9,2}+x_{8,3}+x_{14,1}+x_{13,2}+x_{12,3},~~ {x_{11,1}}+x_{10,2}+x_{9,3}+x_{15,1}+x_{14,2}+x_{13,3},~~ {x_{12,1}}+x_{11,2}+x_{10,3}+x_{16,1}+x_{15,2}+x_{14,3},~~ {x_{13,1}}+x_{12,2}+x_{11,3}+x_{17,1}+x_{16,2}+x_{15,3},~~ {x_{14,1}}+x_{13,2}+x_{12,3}+x_{18,1}+x_{17,2}+x_{16,3},~~ {x_{15,1}}+x_{14,2}+x_{13,3}+x_{19,1}+x_{18,2}+x_{17,3},~~ {x_{16,1}}+x_{15,2}+x_{14,3}+x_{20,1}+x_{19,2}+x_{18,3}\}$.
\end{small}
\end{example}
\medskip
\emph{Corollary 2.} If \mbox{$K-D$} divides $K$, then the proposed code is $${\mathfrak{C}}=\{x_{i}+x_{i+m}+\dots+x_{i+(n-1)m}\ |\ {i = {1,2, \dots, m}\}}$$ where  $K-D=m$ and \mbox{$\frac{K}{K-D}=n$}.\\

\begin{example}
\begin{small}
$K=20,\ U=1,\ D=16$.
\\
\\
$K=20,\ \Delta=15$, capacity=$\frac{2}{5}$.\\
Let $y_{i}=x_{i,1}+x_{i-1,\ 2}$ for $i=1,2,\dots,20.$
\\
\\
$y_{1}=x_{1,1}+x_{20,2}, ~~ y_{11}=x_{11,1}+x_{10,2}$,\\
$y_{2}=x_{2,1}+x_{1,2},  ~~~y_{12}=x_{12,1}+x_{11,2}$,\\
$y_{3}=x_{3,1}+x_{2,2},  ~~~~y_{13}=x_{13,1}+x_{12,2}$,\\
$y_{4}=x_{4,1}+x_{3,2},  ~~~~y_{14}=x_{14,1}+x_{13,2}$,\\
$y_{5}=x_{5,1}+x_{4,2},  ~~~~y_{15}=x_{15,1}+x_{14,2}$,\\
$y_{6}=x_{6,1}+x_{5,2},  ~~~~y_{16}=x_{16,1}+x_{15,2}$,\\
$y_{7}=x_{7,1}+x_{6,2},  ~~~~y_{17}=x_{17,1}+x_{16,2}$,\\
$y_{8}=x_{8,1}+x_{7,2},  ~~~~y_{18}=x_{18,1}+x_{17,2}$,\\
$y_{9}=x_{9,1}+x_{8,2},  ~~~~y_{19}=x_{19,1}+x_{18,2}$,\\
$y_{10}=x_{10,1}+x_{9,2},  ~~y_{20}=x_{20,1}+x_{19,2}$.\\ 
\medskip

The proposed code is,\\
$\mathfrak{C}=\{y_{1}+y_{6}+y_{11}+y_{16}, \\
~~~~~~~~ y_{2}+y_{7}+y_{12}+y_{17}, \\
~~~~~~~~ y_{3}+y_{8}+y_{13}+y_{18}, \\
~~~~~~~~ y_{4}+y_{9}+y_{14}+y_{19},\\
~~~~~~~~ y_{5}+y_{10}+y_{15}+y_{20}\}.$ \\ \\
$ ~ \mathfrak{C^{(2)}}=\{{x_{1,1}}+x_{20,2}+x_{6,1}+x_{5,2}+x_{11,1}+x_{10,2}+x_{16,1}+x_{15,2},\
~~ {x_{2,1}}+x_{1,2}+x_{7,1}+x_{6,2}+x_{12,1}+x_{11,2}+x_{17,1}+x_{16,2},\
~~ {x_{3,1}}+x_{2,2}+x_{8,1}+x_{7,2}+x_{13,1}+x_{12,2}+x_{18,1}+x_{17,2},\
~~ {x_{4,1}}+x_{3,2}+x_{9,1}+x_{8,2}+x_{14,1}+x_{13,2}+x_{19,1}+x_{18,2},\
~~ {x_{5,1}}+x_{4,2}+x_{10,1}+x_{9,2}+x_{15,1}+x_{14,2}+x_{20,1}+x_{19,2}\}.$\\
\end{small}
\end{example}
\medskip
\emph{Corollary 3.} For the case   $\frac{K}{2}-D$ divides $D$, the proposed scalar linear code is,\\ 
$\mathfrak{C}=\{x_{i}+x_{i+m}+\dots+x_{i+pm}, \\ 
~~~~~~~~ x_{i+m}+x_{i+2m}+\dots+x_{i+(p+1)m}, \\
~~~~~~~~~~~~~~~~~~~~ \vdots \\
~~~~~~~~ x_{i+m(p+1)}+x_{i+m(p+2)}+\dots+x_{i+(n-1)m} 
 | i=1,2, \dots ,m\},$
 where  $m=\frac{K}{2}-D$, $n=\frac{K}{\frac{K}{2}-D}$ and $\frac{D}{\frac{K}{2}-D}=p$. \\
\medskip
\begin{example}
\begin{small}
$K=20,\ U=1,\ D=9$. 
\\
\\
$K=20,\ \Delta=8$, capacity=$\frac{2}{12}$.\\
Let $y_{i}=x_{i,1}+x_{i-1,\ 2}$ for $i=1,2,\dots,20.$
\\
\\
$y_{1}=x_{1,1}+x_{20,2}, ~~~y_{2}=x_{2,1}+x_{1,2}, ~~~~y_{3}=x_{3,1}+x_{2,2}$,\\
$y_{4}=x_{4,1}+x_{3,2}, ~~~~y_{5}=x_{5,1}+x_{4,2}, ~~~~y_{6}=x_{6,1}+x_{5,2},$\\
$y_{7}=x_{7,1}+x_{6,2}, ~~~~y_{8}=x_{8,1}+x_{7,2}, ~~~~y_{9}=x_{9,1}+x_{8,2}, $\\
$y_{10}=x_{10,1}+x_{9,2}, ~~~y_{11}=x_{11,1}+x_{10,2},~y_{12}=x_{12,1}+x_{11,2},$\\
$y_{13}=x_{13,1}+x_{12,2}, ~~y_{14}=x_{14,1}+x_{13,2}, ~y_{15}=x_{15,1}+x_{14,2}$,\\
$y_{16}=x_{16,1}+x_{15,2}, ~~y_{17}=x_{17,1}+x_{16,2}, ~y_{18}=x_{18,1}+x_{17,2}$,\\
$y_{19}=x_{19,1}+x_{18,2}, ~y_{20}=x_{20,1}+x_{19,2}$.\\ \\
The proposed code is $\mathfrak{C}=$ \\
$\{y_{1}+y_{3}+y_{5}+y_{7}+y_{9}, ~~ y_{2}+y_{4}+y_{6}+y_{8}+y_{10}, \\
y_{3}+y_{5}+y_{7}+y_{9}+y_{11}, ~~ y_{4}+y_{6}+y_{8}+y_{10}+y_{12},\\ 
y_{5}+y_{7}+y_{9}+y_{11}+y_{13}, ~~ y_{6}+y_{8}+y_{10}+y_{12}+y_{14},\\ 
y_{7}+y_{9}+y_{11}+y_{13}+y_{15}, ~~ y_{8}+y_{10}+y_{12}+y_{14}+y_{16},\\
y_{9}+y_{11}+y_{13}+y_{15}+y_{17}, ~~ y_{10}+y_{12}+y_{14}+y_{16}+y_{18},\\
y_{11}+y_{13}+y_{15}+y_{17}+y_{19}, ~~ y_{12}+y_{14}+y_{16}+y_{18}+y_{20}\}$.\\ 

$ ~ \mathfrak{C^{(2)}}=\{~~{x_{1,1}}+x_{20,2}+x_{3,1}+x_{2,2}+x_{5,1}+x_{4,2}+x_{7,1}+x_{6,2}+x_{9,1}+x_{8,2},\ 
~~ {x_{2,1}}+x_{1,2}+x_{4,1}+x_{3,2}+x_{6,1}+x_{5,2}+x_{8,1}+x_{7,2}+x_{10,1}+x_{9,2},\
~~ {x_{3,1}}+x_{2,2}+x_{5,1}+x_{4,2}+x_{7,1}+x_{6,2}+x_{9,1}+x_{8,2}+x_{11,1}+x_{10,2},\
~~ {x_{4,1}}+x_{3,2}+x_{6,1}+x_{5,2}+x_{8,1}+x_{7,2}+x_{10,1}+x_{9,2}+x_{12,1}+x_{11,2},\
~~ {x_{5,1}}+x_{4,2}+x_{7,1}+x_{6,2}+x_{9,1}+x_{8,2}+x_{11,1}+x_{10,2}+x_{13,1}+x_{12,2},\
~~ {x_{6,1}}+x_{5,2}+x_{8,1}+x_{7,2}+x_{10,1}+x_{9,2}+x_{12,1}+x_{11,2}+x_{14,1}+x_{13,2},\
~~ {x_{7,1}}+x_{6,2}+x_{9,1}+x_{8,2}+x_{11,1}+x_{10,2}+x_{13,1}+x_{12,2}+x_{15,1}+x_{14,2},\
~~ {x_{8,1}}+x_{7,2}+x_{10,1}+x_{9,2}+x_{12,1}+x_{11,2}+x_{14,1}+x_{13,2}+x_{16,1}+x_{15,2}, \
~~ {x_{9,1}}+x_{8,2}+x_{11,1}+x_{10,2}+x_{13,1}+x_{12,2}+x_{15,1}+x_{14,2}+x_{17,1}+x_{16,2},\  
~~ {x_{10,1}}+x_{9,2}+x_{12,1}+x_{11,2}+x_{14,1}+x_{13,2}+x_{16,1}+x_{15,2}+x_{18,1}+x_{17,2},\ 
~~ {x_{11,1}}+x_{10,2}+x_{13,1}+x_{12,2}+x_{15,1}+x_{14,2}+x_{17,1}+x_{16,2}+x_{19,1}+x_{18,2},\ 
~~ {x_{12,1}}+x_{11,2}+x_{14,1}+x_{13,2}+x_{16,1}+x_{15,2}+x_{18,1}+x_{17,2}+x_{20,1}+x_{19,2}\}.$
\end{small}
\end{example}
\medskip
\emph{Corollary 4.} If  $D$ divides $K-\lambda$ and $\lambda$ divides $D$,  then the scalar linear code is given by $\mathfrak{C}=\{{x_{i+(j-1)D}+x_{i+jD}}|\ {i = {1,2,\dots,D}},\ {j = {1,2,\dots,n-1}}\}\\ \cup \{{x_{K-\lambda+r}+x_{K-\lambda+r-\lambda}+\dots+x_{K-\lambda+r-t\lambda}}|\ {r = {1,2,\dots,\lambda}},\\ \ t = 1,2,\dots,\frac{D}{\lambda}$\} for $\frac{K-\lambda}{D}>1$  and $\frac{K-\lambda}{D}=n$.
\medskip
\medskip
\begin{example}
$K=21,\  U=2,\ D=6$
\\
\\
$K=21,\Delta=4, \lambda=1$, capacity=$\frac{3}{17}$.\\
Let $y_{i}=x_{i,1}+x_{i-1,\ 2}+x_{i-2,\ 3}$ for $i=1,2,\dots,21.$
\\
\\
\begin{small}
$y_{1}=x_{1,1}+x_{21,2}+x_{20,3}, ~~ y_{11}=x_{11,1}+x_{10,2}+x_{9,3}$,\\
$y_{2}=x_{2,1}+x_{1,2}+x_{21,3},  ~~~y_{12}=x_{12,1}+x_{11,2}+x_{10,3}$,\\
$y_{3}=x_{3,1}+x_{2,2}+x_{1,3},  ~~~~y_{13}=x_{13,1}+x_{12,2}+x_{11,3}$,\\
$y_{4}=x_{4,1}+x_{3,2}+x_{2,3},  ~~~~y_{14}=x_{14,1}+x_{13,2}+x_{12,3}$,\\
$y_{5}=x_{5,1}+x_{4,2}+x_{3,3},  ~~~~y_{15}=x_{15,1}+x_{14,2}+x_{13,3}$,\\
$y_{6}=x_{6,1}+x_{5,2}+x_{4,3},  ~~~~y_{16}=x_{16,1}+x_{15,2}+x_{14,3}$,\\
$y_{7}=x_{7,1}+x_{6,2}+x_{5,3},  ~~~~y_{17}=x_{17,1}+x_{16,2}+x_{15,3}$,\\
$y_{8}=x_{8,1}+x_{7,2}+x_{6,3},  ~~~~y_{18}=x_{18,1}+x_{17,2}+x_{16,3}$,\\
$y_{9}=x_{9,1}+x_{8,2}+x_{7,3},  ~~~~y_{19}=x_{19,1}+x_{18,2}+x_{17,3}$,\\
$y_{10}=x_{10,1}+x_{9,2}+x_{8,3},  ~~y_{20}=x_{20,1}+x_{19,2}+x_{18,3}$,\\
and $y_{21}=x_{21,1}+x_{20,2}+x_{19,3}.$\\ 

The proposed code is $\mathfrak{C}=\{y_{1}+y_{5}, ~y_{2}+y_{6}, ~y_{3}+y_{7}, ~y_{4}+y_{8}, ~y_{5}+y_{9}, ~y_{6}+y_{10}, ~y_{7}+y_{11}, ~y_{8}+y_{12}, ~y_{9}+y_{13}, ~y_{10}+y_{14}, ~y_{11}+y_{15}, ~y_{12}+y_{16}, ~y_{13}+y_{17}, ~y_{14}+y_{18}, ~y_{15}+y_{19}, ~y_{16}+y_{20}, ~y_{17}+y_{18}+y_{19}+y_{20}+y_{21}\}.$\\ \\
 $ ~ \mathfrak{C^{(3)}}=\{{x_{1,1}}+x_{21,2}+x_{20,3}+x_{5,1}+x_{4,2}+x_{3,3},\
~~ {x_{2,1}}+x_{1,2}+x_{21,3}+x_{6,1}+x_{5,2}+x_{4,3},\
~~ {x_{3,1}}+x_{2,2}+x_{1,3}+x_{7,1}+x_{6,2}+x_{5,3},\
~~ {x_{4,1}}+x_{3,2}+x_{2,3}+x_{8,1}+x_{7,2}+x_{6,3},\
~~ {x_{5,1}}+x_{4,2}+x_{3,3}+x_{9,1}+x_{8,2}+x_{7,3},\
~~ {x_{6,1}}+x_{5,2}+x_{4,3}+x_{10,1}+x_{9,2}+x_{8,3},\
~~ {x_{7,1}}+x_{6,2}+x_{5,3}+x_{11,1}+x_{10,2}+x_{9,3},\
~~ {x_{8,1}}+x_{7,2}+x_{6,3}+x_{12,1}+x_{11,2}+x_{10,3},\
~~ {x_{9,1}}+x_{8,2}+x_{7,3}+x_{13,1}+x_{12,2}+x_{11,3},\
~~ {x_{10,1}}+x_{9,2}+x_{8,3}+x_{14,1}+x_{13,2}+x_{12,3},\
~~ {x_{11,1}}+x_{10,2}+x_{9,3}+x_{15,1}+x_{14,2}+x_{13,3},\
~~ {x_{12,1}}+x_{11,2}+x_{10,3}+x_{16,1}+x_{15,2}+x_{14,3},\
~~ {x_{13,1}}+x_{12,2}+x_{11,3}+x_{17,1}+x_{16,2}+x_{15,3},\
~~ {x_{14,1}}+x_{13,2}+x_{12,3}+x_{18,1}+x_{17,2}+x_{16,3},\
~~ {x_{15,1}}+x_{14,2}+x_{13,3}+x_{19,1}+x_{18,2}+x_{17,3},\
~~ {x_{16,1}}+x_{15,2}+x_{14,3}+x_{20,1}+x_{19,2}+x_{18,3},\
~~ {x_{17,1}}+x_{16,2}+x_{15,3}+x_{18,1}+x_{17,2}+x_{16,3}+x_{19,1}+x_{18,2}+x_{17,3}+x_{20,1}+x_{19,2}+x_{18,3}+x_{21,1}+x_{20,2}+x_{19,3}\}.$
 \end{small}
\end{example}
\medskip
\emph{Corollary 5.} If $K-D$ divides $K-\lambda,$  $\lambda$ divides $(K-D)$, then the proposed scalar linear code is $\mathfrak{C}=\{x_{i}+x_{i+m}+\dots+x_{i+(q-1)m}+x_{qm+1+(i-1) mod \lambda)}
|\ i = {1,2,\dots,m}\},$ where $K-D=m$, and $\frac{K-\lambda}{K-D}=q$.
\medskip
\medskip
\begin{example}
$K=21,\ U=1,\ D=17$.
\\
\\
$K=21,\ \Delta=16,$ \ capacity=$\frac{2}{5}$.\\
Let $y_{i}=x_{i,1}+x_{i-1,\ 2}$ for $i=1,2,\dots,21.$
\\
\\
\begin{small}
$y_{1}=x_{1,1}+x_{21,2}, ~~ y_{11}=x_{11,1}+x_{10,2}$,\\
$y_{2}=x_{2,1}+x_{1,2},  ~~~y_{12}=x_{12,1}+x_{11,2}$,\\
$y_{3}=x_{3,1}+x_{2,2},  ~~~~y_{13}=x_{13,1}+x_{12,2}$,\\
$y_{4}=x_{4,1}+x_{3,2},  ~~~~y_{14}=x_{14,1}+x_{13,2}$,\\
$y_{5}=x_{5,1}+x_{4,2},  ~~~~y_{15}=x_{15,1}+x_{14,2}$,\\
$y_{6}=x_{6,1}+x_{5,2},  ~~~~y_{16}=x_{16,1}+x_{15,2}$,\\
$y_{7}=x_{7,1}+x_{6,2},  ~~~~y_{17}=x_{17,1}+x_{16,2}$,\\
$y_{8}=x_{8,1}+x_{7,2},  ~~~~y_{18}=x_{18,1}+x_{17,2}$,\\
$y_{9}=x_{9,1}+x_{8,2},  ~~~~y_{19}=x_{19,1}+x_{18,2}$,\\
$y_{10}=x_{10,1}+x_{9,2},  ~~y_{20}=x_{20,1}+x_{19,2}$,\\
and $y_{21}=x_{21,1}+x_{20,2}.$\\ \\
The proposed code is $\mathfrak{C}=\{y_{1}+y_{6}+y_{11}+y_{16}+y_{21}, \
 y_{2}+y_{7}+y_{12}+y_{17}+y_{21}, \
 y_{3}+y_{8}+y_{13}+y_{18}+y_{21}, \
 y_{4}+y_{9}+y_{14}+y_{19}+y_{21},\
 y_{5}+y_{10}+y_{15}+y_{20}+y_{21}\}.$\\ \\
$ ~ \mathfrak{C^{(2)}}=\{{x_{1,1}}+x_{21,2}+x_{6,1}+x_{5,2}+x_{11,1}+x_{10,2}+x_{16,1}+x_{15,2}+x_{21,1}+x_{20,2},\ 
~~ {x_{2,1}}+x_{1,2}+x_{7,1}+x_{6,2}+x_{12,1}+x_{11,2}+x_{17,1}+x_{16,2}+x_{21,1}+x_{20,2},\ 
~~ {x_{3,1}}+x_{2,2}+x_{8,1}+x_{7,2}+x_{13,1}+x_{12,2}+x_{18,1}+x_{17,2}+x_{21,1}+x_{20,2},\ 
~~ {x_{4,1}}+x_{3,2}+x_{9,1}+x_{8,2}+x_{14,1}+x_{13,2}+x_{19,1}+x_{18,2}+x_{21,1}+x_{20,2},\ 
~~ {x_{5,1}}+x_{4,2}+x_{10,1}+x_{9,2}+x_{15,1}+x_{14,2}+x_{20,1}+x_{19,2}+x_{21,1}+x_{20,2}\}.$\\
\end{small}
\end{example}
\medskip
\medskip
\emph{Corollary 6.} For the case  $D+\lambda$ divides $K,$  $\lambda$ divides $D$, the scalar linear code is\\
$\mathfrak{C}=\{x_{i+j\lambda}+x_{i+(j+1)\lambda}+\dots+x_{i+(j+p)\lambda}|\ {i =1,2,\dots,\lambda},\\ j= 1,2,\dots,\frac{K-D-\lambda}{\lambda}$\}
where  $\frac{D}{\lambda}=p$ and $\frac{K}{D+\lambda}=n.$
\medskip
\medskip
\begin{example}
$K=18,\ U=1,\ D=6.$ 
\\
\\
$K=18,\ \Delta=5, \lambda=1$, capacity=$\frac{2}{13}$. \\
Let $y_{i}=x_{i,1}+x_{i-1,\ 2}$ for $i=1,2,\dots,18.$ 
\\
\\
\begin{small}
$y_{1}=x_{1,1}+x_{18,2}, ~~ y_{10}=x_{10,1}+x_{9,2},$,\\
$y_{2}=x_{2,1}+x_{1,2},  ~~~y_{11}=x_{11,1}+x_{10,2}$,\\
$y_{3}=x_{3,1}+x_{2,2},  ~~~~y_{12}=x_{12,1}+x_{11,2}$,\\
$y_{4}=x_{4,1}+x_{3,2},  ~~~~y_{13}=x_{13,1}+x_{12,2}$,\\
$y_{5}=x_{5,1}+x_{4,2},  ~~~~y_{14}=x_{14,1}+x_{13,2}$,\\
$y_{6}=x_{6,1}+x_{5,2},  ~~~~y_{15}=x_{15,1}+x_{14,2}$,\\
$y_{7}=x_{7,1}+x_{6,2},  ~~~~y_{16}=x_{16,1}+x_{15,2}$,\\
$y_{8}=x_{8,1}+x_{7,2},  ~~~~y_{17}=x_{17,1}+x_{16,2}$,\\
$y_{9}=x_{9,1}+x_{8,2},  ~~~~y_{18}=x_{18,1}+x_{17,2}$,\\ \\
The proposed code is $\mathfrak{C}=\{y_{1}+y_{2}+y_{3}+y_{4}+y_{5}+y_{6},~ y_{2}+y_{3}+y_{4}+y_{5}+y_{6}+y_{7}, \
y_{3}+y_{4}+y_{5}+y_{6}+y_{7}+y_{8}, ~ y_{4}+y_{5}+y_{6}+y_{7}+y_{8}+y_{9}, \
y_{5}+y_{6}+y_{7}+y_{8}+y_{9}+y_{10}, ~ y_{6}+y_{7}+y_{8}+y_{9}+y_{10}+y_{11}, \
y_{7}+y_{8}+y_{9}+y_{10}+y_{11}+y_{12}, ~ y_{8}+y_{9}+y_{10}+y_{11}+y_{12}+y_{13}, \
y_{9}+y_{10}+y_{11}+y_{12}+y_{13}+y_{14}, ~ y_{10}+y_{11}+y_{12}+y_{13}+y_{14}+y_{15}, \
y_{11}+y_{12}+y_{13}+y_{14}+y_{15}+y_{16}, ~ y_{12}+y_{13}+y_{14}+y_{15}+y_{16}+y_{17}, \
y_{13}+y_{14}+y_{15}+y_{16}+y_{17}+y_{18}\}$.\\ \\

{
$ ~ \mathfrak{C^{(2)}}=\{{x_{1,1}}+x_{18,2}+x_{2,1}+x_{1,2}+x_{3,1}+x_{2,2}+x_{4,1}+x_{3,2}+x_{5,1}+x_{4,2}+x_{6,1}+x_{5,2},\ 
~~ {x_{2,1}}+x_{1,2}+x_{3,1}+x_{2,2}+x_{4,1}+x_{3,2}+x_{5,1}+x_{4,2}+x_{6,1}+x_{5,2}+x_{7,1}+x_{6,2},\ 
~~ {x_{3,1}}+x_{2,2}+x_{4,1}+x_{3,2}+x_{5,1}+x_{4,2}+x_{6,1}+x_{5,2}+x_{7,1}+x_{6,2}+x_{8,1}+x_{7,2},\ 
~~ {x_{4,1}}+x_{3,2}+x_{5,1}+x_{4,2}+x_{6,1}+x_{5,2}+x_{7,1}+x_{6,2}+x_{8,1}+x_{7,2}+x_{9,1}+x_{8,2},\ 
~~ {x_{5,1}}+x_{4,2}+x_{6,1}+x_{5,2}+x_{7,1}+x_{6,2}+x_{8,1}+x_{7,2}+x_{9,1}+x_{8,2}+x_{10,1}+x_{9,2},\ 
~~ {x_{6,1}}+x_{5,2}+x_{7,1}+x_{6,2}+x_{8,1}+x_{7,2}+x_{9,1}+x_{8,2}+x_{10,1}+x_{9,2}+x_{11,1}+x_{10,2},\ 
~~ {x_{7,1}}+x_{6,2}+x_{8,1}+x_{7,2}+x_{9,1}+x_{8,2}+x_{10,1}+x_{9,2}+x_{11,1}+x_{10,2}+x_{12,1}+x_{11,2},\ 
~~ {x_{8,1}}+x_{7,2}+x_{9,1}+x_{8,2}+x_{10,1}+x_{9,2}+x_{11,1}+x_{10,2}+x_{12,1}+x_{11,2}+x_{13,1}+x_{12,2},\ 
~~ {x_{9,1}}+x_{8,2}+x_{10,1}+x_{9,2}+x_{11,1}+x_{10,2}+x_{12,1}+x_{11,2}+x_{13,1}+x_{12,2}+x_{14,1}+x_{13,2},\ 
~~ {x_{10,1}}+x_{9,2}+x_{11,1}+x_{10,2}+x_{12,1}+x_{11,2}+x_{13,1}+x_{12,2}+x_{14,1}+x_{13,2}+x_{15,1}+x_{14,2},\ 
~~ {x_{11,1}}+x_{10,2}+x_{12,1}+x_{11,2}+x_{13,1}+x_{12,2}+x_{14,1}+x_{13,2}+x_{15,1}+x_{14,2}+x_{16,1}+x_{15,2},\ 
~~ {x_{12,1}}+x_{11,2}+x_{13,1}+x_{12,2}+x_{14,1}+x_{13,2}+x_{15,1}+x_{14,2}+x_{16,1}+x_{15,2}+x_{17,1}+x_{16,2},\ 
~~ {x_{13,1}}+x_{12,2}+x_{14,1}+x_{13,2}+x_{15,1}+x_{14,2}+x_{16,1}+x_{15,2}+x_{17,1}+x_{16,2}+x_{18,1}+x_{17,2}\}.$
}
\end{small}
\end{example}
\medskip
\emph{Corollary 7.} If \mbox{$K-D+\lambda$} divides $K$ and $\lambda$ divides \mbox{$K-D$}, then the scalar linear code is given by $\mathfrak{C}=\{x_{i}+x_{i+\lambda}+x_{i+\lambda+(K-D)}+x_{i+2\lambda+(K-D)}+x_{i+2\lambda+2(K-D)}+x_{i+3\lambda+2(K-D)}+\dots+x_{i+(p-1)\lambda+(p-1)(K-D)}+x_{i+p\lambda+(p-1)(K-D)}|\ $\mbox{$i= \{1,2,\dots,K-D$\}}\} where $\frac{K}{K-D+\lambda}=p$ and $\frac{K-D}{\lambda}=m$.
\medskip
\medskip
\begin{example}
$K=24,\ U=1,\ D=20.$
\\
\\
$K=24,\ \Delta=19, \lambda=1$, capacity=$\frac{2}{5}$.\\
Let $y_{i}=x_{i,1}+x_{i-1,\ 2}$ for $i=1,2,\dots,24.$\\
\\
\\
\begin{small}
$y_{1}=x_{1,1}+x_{24,2}, ~~ ~y_{13}=x_{13,1}+x_{12,2}$,\\
$y_{2}=x_{2,1}+x_{1,2},  ~~~~y_{14}=x_{14,1}+x_{13,2}$,\\
$y_{3}=x_{3,1}+x_{2,2},  ~~~~y_{15}=x_{15,1}+x_{14,2}$,\\
$y_{4}=x_{4,1}+x_{3,2},  ~~~~y_{16}=x_{16,1}+x_{15,2}$,\\
$y_{5}=x_{5,1}+x_{4,2},  ~~~~y_{17}=x_{17,1}+x_{16,2}$,\\
$y_{6}=x_{6,1}+x_{5,2},  ~~~~y_{18}=x_{18,1}+x_{17,2}$,\\
$y_{7}=x_{7,1}+x_{6,2},  ~~~~y_{19}=x_{19,1}+x_{18,2}$,\\
$y_{8}=x_{8,1}+x_{7,2},  ~~~~y_{20}=x_{20,1}+x_{19,2}$,\\
$y_{9}=x_{9,1}+x_{8,2},  ~~~~y_{21}=x_{21,1}+x_{20,2}$,\\
$y_{10}=x_{10,1}+x_{9,2},  ~~y_{22}=x_{22,1}+x_{21,2}$,\\
$y_{11}=x_{11,1}+x_{10,2},  ~y_{23}=x_{23,1}+x_{22,2}$,\\
$y_{12}=x_{12,1}+x_{11,2},  ~y_{24}=x_{24,1}+x_{23,2}$.
\\
\\
The proposed code is $\mathfrak{C}=\{y_{1}+y_{2}+ y_{7}+y_{8}+y_{13}+y_{14}+y_{19}+y_{20},\ y_{2}+y_{3}+ y_{8}+y_{9}+y_{14}+y_{15}+y_{20}+y_{21},\ y_{3}+y_{4}+ y_{9}+y_{10}+y_{15}+y_{16}+y_{21}+y_{22},\ y_{4}+y_{5}+ y_{10}+y_{11}+y_{16}+y_{17}+y_{22}+y_{23},\ y_{5}+y_{6}+ y_{11}+y_{12}+y_{17}+y_{18}+y_{23}+y_{24}\}$.
\\
\\
$ \mathfrak{C^{(2)}}=\{{x_{1,1}}+x_{24,2}+x_{2,1}+x_{1,2}+x_{7,1}+x_{6,2}+x_{8,1}+x_{7,2}+x_{13,1}+x_{12,2}+x_{14,1}+x_{13,2}+x_{19,1}+x_{18,2}+x_{20,1}+x_{19,2},\ 
~~ {x_{2,1}}+x_{1,2}+x_{3,1}+x_{2,2}+x_{8,1}+x_{7,2}+x_{9,1}+x_{8,2}+x_{14,1}+x_{13,2}+x_{15,1}+x_{14,2}+x_{20,1}+x_{19,2}+x_{21,1}+x_{20,2},\ 
~~ {x_{3,1}}+x_{2,2}+x_{4,1}+x_{3,2}+x_{9,1}+x_{8,2}+x_{10,1}+x_{9,2}+x_{15,1}+x_{14,2}+x_{16,1}+x_{15,2}+x_{21,1}+x_{20,2}+x_{22,1}+x_{21,2},\ 
~~ {x_{4,1}}+x_{3,2}+x_{5,1}+x_{4,2}+x_{10,1}+x_{9,2}+x_{11,1}+x_{10,2}+x_{16,1}+x_{15,2}+x_{17,1}+x_{16,2}+x_{22,1}+x_{21,2}+x_{23,1}+x_{22,2},\ 
~~ {x_{5,1}}+x_{4,2}+x_{6,1}+x_{5,2}+x_{11,1}+x_{10,2}+x_{12,1}+x_{11,2}+x_{17,1}+x_{16,2}+x_{18,1}+x_{17,2}+x_{23,1}+x_{22,2}+x_{24,1}+x_{23,2}\}.$\\
\end{small}
\end{example}
\medskip
\emph{Corollary 8.} If  $D$ divides $K+\lambda$ and $\lambda$ divides $D$, then the scalar linear code is given by $\mathfrak{C}=\{{x_{i+(j-1)D}+x_{i+jD}}|\ i = \{1,2,\dots,D\},\ j = \{1,2,\dots,n-2\}\}\\ \cup \{x_{K-2D+1+\lambda+i'}+x_{K-D+1+i'}+x_{K-\lambda+1+i' mod \lambda}|\ i' = \{0,1,2,\dots,p-1\}\}$ where $\frac{K+\lambda}{D}=n(>2)$, $p=K$ mod $D$ = $D-\lambda$.
\medskip
\medskip
\begin{example}
$K=19, U=2,\ D=7.$
\\
\\
$K=19,\ \Delta=5,\ \lambda=1$, capacity=$\frac{3}{14}$.\\
Let $y_{i}=x_{i,1}+x_{i-1,\ 2}+x_{i-2,\ 3}$ for $i=1,2,\dots,19.$
\\
\\
\begin{small}
$y_{1}=x_{1,1}+x_{19,2}+x_{18,3}, ~~ y_{11}=x_{11,1}+x_{10,2}+x_{9,3}$,\\
$y_{2}=x_{2,1}+x_{1,2}+x_{20,3},  ~~~y_{12}=x_{12,1}+x_{11,2}+x_{10,3}$,\\
$y_{3}=x_{3,1}+x_{2,2}+x_{1,3},  ~~~~y_{13}=x_{13,1}+x_{12,2}+x_{11,3}$,\\
$y_{4}=x_{4,1}+x_{3,2}+x_{2,3},  ~~~~y_{14}=x_{14,1}+x_{13,2}+x_{12,3}$,\\
$y_{5}=x_{5,1}+x_{4,2}+x_{3,3},  ~~~~y_{15}=x_{15,1}+x_{14,2}+x_{13,3}$,\\
$y_{6}=x_{6,1}+x_{5,2}+x_{4,3},  ~~~~y_{16}=x_{16,1}+x_{15,2}+x_{14,3}$,\\
$y_{7}=x_{7,1}+x_{6,2}+x_{5,3},  ~~~~y_{17}=x_{17,1}+x_{16,2}+x_{15,3}$,\\
$y_{8}=x_{8,1}+x_{7,2}+x_{6,3},  ~~~~y_{18}=x_{18,1}+x_{17,2}+x_{16,3}$,\\
$y_{9}=x_{9,1}+x_{8,2}+x_{7,3},  ~~~~y_{19}=x_{19,1}+x_{18,2}+x_{17,3}$,\\
$y_{10}=x_{10,1}+x_{9,2}+x_{8,3}.$\\ \\
The proposed code is $\mathfrak{C}=\{y_{1}+y_{6},\ y_{6}+y_{11},\ y_{11}+y_{15}+y_{19},\ y_{2}+y_{7},\ y_{7}+y_{12},\ y_{12}+y_{16}+y_{19},\ y_{3}+y_{8},\ y_{8}+y_{13},\ y_{13}+y_{17}+y_{19},\ y_{4}+y_{9},\ y_{9}+y_{14},\ y_{14}+y_{18}+y_{19}, \ y_{5}+y_{10},\ y_{10}+y_{15}\}.$
{
$ ~ \mathfrak{C^{(3)}}=\{{x_{1,1}}+x_{19,2}+x_{18,3}+x_{6,1}+x_{5,2}+x_{4,3},\
~~ {x_{2,1}}+x_{1,2}+x_{20,3}+x_{7,1}+x_{6,2}+x_{5,3},\
~~ {x_{3,1}}+x_{2,2}+x_{1,3}+x_{8,1}+x_{7,2}+x_{6,3},\
~~ {x_{4,1}}+x_{3,2}+x_{2,3}+x_{9,1}+x_{8,2}+x_{7,3},\
~~ {x_{5,1}}+x_{4,2}+x_{3,3}+x_{10,1}+x_{9,2}+x_{8,3},\
~~ {x_{6,1}}+x_{5,2}+x_{4,3}+x_{11,1}+x_{10,2}+x_{9,3},\
~~ {x_{7,1}}+x_{6,2}+x_{5,3}+x_{12,1}+x_{11,2}+x_{10,3},\
~~ {x_{8,1}}+x_{7,2}+x_{6,3}+x_{13,1}+x_{12,2}+x_{11,3},\
~~ {x_{9,1}}+x_{8,2}+x_{7,3}+x_{14,1}+x_{13,2}+x_{12,3},\
~~ {x_{10,1}}+x_{9,2}+x_{8,3}+x_{15,1}+x_{14,2}+x_{13,3},\
~~ {x_{11,1}}+x_{10,2}+x_{9,3}+x_{15,1}+x_{14,2}+x_{13,3}+x_{19,1}+x_{18,2}+x_{17,3},\
~~ {x_{12,1}}+x_{11,2}+x_{10,3}+x_{16,1}+x_{15,2}+x_{14,3}+x_{19,1}+x_{18,2}+x_{17,3},\
~~ {x_{13,1}}+x_{12,2}+x_{11,3}+x_{17,1}+x_{16,2}+x_{15,3}+x_{19,1}+x_{18,2}+x_{17,3},\
~~ {x_{14,1}}+x_{13,2}+x_{12,3}+x_{18,1}+x_{17,2}+x_{16,3}+x_{19,1}+x_{18,2}+x_{17,3}\}.$
}
\end{small}
\end{example}
\medskip
\emph{Corollary 9.} If $K-D$ divides $K+\lambda$ and $\lambda$ divides $K-D$, then the scalar linear code\\ $\mathfrak{C}$=$\{x_{k}+x_{k+m}+x_{k+2m}+ \dots+x_{k+(q-1)m}+x_{k+(q-1)m+\lambda}+x_{k+(q-1)m+2\lambda}+ \dots + x_{k+(q-1)m+(s-2)\lambda}| k =1,2,\dots,\lambda\}\cup \{x_{k}+x_{k+m}+x_{k+2m}\dots+x_{k+(q-2)m}+x_{k+(q-1)m-\lambda}|\ k = \lambda+1,\lambda+2,\dots,p\}\cup \{x_{k}+x_{k+m}+x_{k+2m}+ \dots+x_{k+(q-2)m}+x_{k+(q-2)m+\lambda}+x_{k+(q-2)m+2\lambda}+\dots + x_{k+(q-2)m+(s-1)\lambda}|\ k =p+1,p+2,\dots,m\}$\\ where $K-D=m$, $K-D-\lambda=p$, $\frac{K+\lambda}{K-D}=q$  and \mbox{$\frac{K-D}{\lambda}=s$}.
\medskip
\begin{example}
$K=28,\ U=1,\ D=19.$
\\
\\
$K=28,\ \Delta=18,\lambda=2$, capacity=$\frac{2}{10}$.\\
Let $y_{i}=x_{i,1}+x_{i-1,\ 2}$ for $i=1,2,\dots,28.$
\\
\\
$y_{1}=x_{1,1}+x_{28,2}, ~~ ~y_{15}=x_{15,1}+x_{14,2}$,\\
$y_{2}=x_{2,1}+x_{1,2},  ~~~~y_{16}=x_{16,1}+x_{15,2}$,\\
$y_{3}=x_{3,1}+x_{2,2},  ~~~~y_{17}=x_{17,1}+x_{16,2}$,\\
$y_{4}=x_{4,1}+x_{3,2},  ~~~~y_{18}=x_{18,1}+x_{17,2}$,\\
$y_{5}=x_{5,1}+x_{4,2},  ~~~~y_{19}=x_{19,1}+x_{18,2}$,\\
$y_{6}=x_{6,1}+x_{5,2},  ~~~~y_{20}=x_{20,1}+x_{19,2}$,\\
$y_{7}=x_{7,1}+x_{6,2},  ~~~~y_{21}=x_{21,1}+x_{20,2}$,\\
$y_{8}=x_{8,1}+x_{7,2},  ~~~~y_{22}=x_{22,1}+x_{21,2}$,\\
$y_{9}=x_{9,1}+x_{8,2},  ~~~~y_{23}=x_{23,1}+x_{22,2}$,\\
$y_{10}=x_{10,1}+x_{9,2},  ~~y_{24}=x_{24,1}+x_{23,2}$,\\
$y_{11}=x_{11,1}+x_{10,2},  ~y_{25}=x_{25,1}+x_{24,2}$,\\
$y_{12}=x_{12,1}+x_{11,2},  ~y_{26}=x_{26,1}+x_{25,2}$,\\
$y_{13}=x_{13,1}+x_{12,2},  ~y_{27}=x_{27,1}+x_{26,2}$,\\
$y_{14}=x_{14,1}+x_{13,2},  ~y_{28}=x_{28,1}+x_{27,2}$.
\\
\\
The proposed code is $\mathfrak{C}=\{y_{1}+y_{11}+y_{21}+y_{23}+y_{25}+y_{27},\ y_{2}+y_{12}+y_{22}+y_{24}+y_{26}+y_{28},\ y_{3}+y_{13}+y_{21},\ y_{4}+y_{14}+y_{22},\ y_{5}+y_{15}+y_{23},\ y_{6}+y_{16}+y_{24},\ y_{7}+y_{17}+y_{25},\ y_{8}+y_{18}+y_{26},\ y_{9}+y_{19}+y_{21}+y_{23}+y_{25}+y_{27},\ y_{10}+y_{20}+y_{22}+y_{24}+y_{26}+y_{28}\}.$
\\
\\
$ \mathfrak{C^{(2)}}=\{{x_{1,1}}+x_{28,2}+x_{11,1}+x_{10,2}+x_{21,1}+x_{20,2}+x_{23,1}+x_{22,2}+x_{25,1}+x_{24,2}+x_{27,1}+x_{26,2},\
~~ {x_{2,1}}+x_{1,2}+x_{12,1}+x_{11,2}+x_{22,1}+x_{21,2}+x_{24,1}+x_{23,2}+x_{26,1}+x_{25,2}+x_{28,1}+x_{27,2},\
~~ {x_{3,1}}+x_{2,2}+x_{13,1}+x_{12,2}+x_{21,1}+x_{20,2},\
~~ {x_{4,1}}+x_{3,2}+x_{14,1}+x_{13,2}+x_{22,1}+x_{21,2},\
~~ {x_{5,1}}+x_{4,2}+x_{15,1}+x_{14,2}+x_{23,1}+x_{22,2},\
~~ {x_{6,1}}+x_{5,2}+x_{16,1}+x_{15,2}+x_{24,1}+x_{23,2},\
~~ {x_{7,1}}+x_{6,2}+x_{17,1}+x_{16,2}+x_{25,1}+x_{24,2},\
~~ {x_{8,1}}+x_{7,2}+x_{18,1}+x_{17,2}+x_{26,1}+x_{25,2},\
~~ {x_{9,1}}+x_{8,2}+x_{19,1}+x_{18,2}+x_{21,1}+x_{20,2}+x_{23,1}+x_{22,2}+x_{25,1}+x_{24,2}+x_{27,1}+x_{26,2},\
~~ {x_{10,1}}+x_{9,2}+x_{20,1}+x_{19,2}+x_{22,1}+x_{21,2}+x_{24,1}+x_{23,2}+x_{26,1}+x_{25,2}+x_{28,1}+x_{27,2}\}.$\\
\end{example}
\section{Discussion}
\label{sec6}
In this paper a construction is given for vector linear index codes of multiple unicast index problems from scalar linear codes which results in a sequence of index coding problems with same number of messages and receivers and two sided antidote patterns. Moreover, it is shown that if the problem with which the construction begins  has an optimal linear index code then it induces an optimal linear index code for the vector index coding problem. This  construction has been used on few classes of index coding problems given in \cite{MRRarXiv} for which optimal linear index codes are known and new classes codes have been obtained starting from these classes of codes. \\

Another interesting direction of further research is to study the suitability of the new classes of codes presented in this paper for application to noise broadcasting problem. Recently, it has been observed that in a noisy index coding problem it is desirable for the purpose of reducing the probability of error that  the receivers use as small a number of transmissions from the source as possible and linear index codes with this property have been reported in \cite{TRCR}, \cite{KaR}. While the report \cite{TRCR} considers fading broadcast channels, in \cite{AnR} AWGN channels are considered and it is reported that linear index codes with minimum length (capacity achieving codes or optimal length codes) help to facilitate to achieve more reduction in probability of error compared to non-minimum length codes for receivers with large amount of side-information. These aspects remain to be investigated for the new classes of sequences of vector codes presented in this paper.

\end{document}